\documentclass[a4paper]{cas-sc}
\usepackage[numbers,longnamesfirst]{natbib}
\usepackage{amsmath}

\usepackage{mathrsfs}%花体字母

\usepackage{amsthm}
\newtheorem{theorem}{Theorem}[section]
\newtheorem{lemma}{Lemma}[section]
\newtheorem{example}{Example}[section]
\newtheorem{definition}{Definition}[section]
\newtheorem{remark}{Remark}[section]
\newtheorem{corollary}{Corollary}[section]
\usepackage{amssymb} %阿拉伯数字
\usepackage{rotating}%反斜线省略号
\begin{document}
\let\WriteBookmarks\relax
\def\floatpagepagefraction{1}
\def\textpagefraction{.001}
\shortauthors{Wenyu Han et~al.}
\shorttitle{}
\title [mode = title]{Extremal Self-Dual Codes and Linear Complementary Dual Codes from Double Circulant Codes}

\author{Wenyu Han}[orcid=0009-0006-1666-2844]
\ead{15692342169@163.com}
\author{Tongjiang Yan}[orcid=0000-0002-9647-503X]
\cormark[1]
\ead{yantoji@163.com}
\author{Ming Yan}[orcid=0009-0001-3275-8321]
\ead{yanming_a1@163.com}
\affiliation{
                addressline={China University of Petroleum (East China)}, 
                city={Qingdao, Shandong},
                postcode={266000}, 
                country={China}}
\cortext[cor1]{Corresponding author}

\begin{abstract}
This paper explores extremal self-dual double circulant (DC) codes and linear complementary dual (LCD) codes of arbitrary length over the Galois field $\mathbb F_2$. We establish the sufficient and necessary conditions for DC codes and bordered DC codes to be self-dual and identify the conditions for self-dual DC codes of length up to 44 to be extremal or non-extremal. Additionally, The self-duality and extremality between DC codes and bordered DC codes are also examined. Finally, sufficient conditions for bordered DC codes to be LCD codes over $\mathbb F_2$ under Euclidean inner product are presented.

\end{abstract}


\maketitle
\section{Introduction}
 Self-dual codes are one of the most interesting classes of linear codes, include many well-known examples such as extended Hamming codes, extended Golay codes, and certain quadratic residue codes. Self-dual codes play a crucial role in the construction of quantum stabilizer codes \cite{ref1} and the determination of weight enumerators \cite{ref2}, making their study a highly significant area of research.

Double circulant (DC) codes, a special class of 1-generator quasi-cyclic codes with index 2, have been extensively studied by researchers since the 1960s \cite{ref3,ref4}. The generator matrices of DC codes are composed of an identity matrix and a circulant matrix, which endow them with favorable algebraic properties \cite{ref4,ref5}. The unique structure of their generator matrices enables the determination of a self-dual double circulant code once an orthogonal circulant matrix is identified. The construction and enumeration of orthogonal circulant matrices have been extensively studied \cite{ref6}\nocite{ref7}-\cite{ref8}. In the majority of research on double circulant codes, scholars have focused on constructing these codes from the perspective of generator matrices, leveraging knowledge of cyclotomic numbers, power residues, and sequences \cite{ref5,ref9,ref10,ref11}. Moreover, considering the importance of minimum distance and weight distribution, researchers have been dedicated to finding extremal self-dual double circulant codes and determining their weight distributions \cite{ref12}\nocite{ref13,ref14,ref15}-\cite{ref16}. However, in the work of predecessors, identifying extremal self-dual double circulant codes has not been a straightforward task. If the construction method of the generator matrix is fixed, the resulting double circulant code may not be extremal. Exhaustive search methods, while capable of finding extremal codes, require considerable computational power. Therefore, the primary objective of this paper is to identify methods that facilitate the determination of extremal self-dual DC codes of length up to 44 and bordered DC codes of length up to 20 in the polynomial form.

Linear complementary dual (LCD) codes are linear codes that intersect with their dual trivially which were introduced by J.L.Massey in 1992 \cite{ref17}. Massey showed that LCD codes provide an optimal linear coding scheme for a two-user binary adder channel.  In recent years, much work has been
done concerning the construction of LCD codes \cite{ref18}\nocite{ref19,ref20,ref21,ref22}-\cite{ref23}. Particularly, in \cite{ref23}, Guan et al. presented the sufficient and necessary conditions for one-generator quasi-cyclic
codes to be LCD codes involving Euclidean, Hermitian, and symplectic inner products. Inspired by their work, this paper establishes a connection between double circulant LCD codes and bordered double circulant LCD codes, and proposes sufficient conditions for bordered double circulant codes to be LCD over $\mathbb F_2$. 

The paper is organized as follows. Section 2 contains some definitions and preliminaries needed thereafter. Section 3 presents conditions for double circulant self-dual codes of length up to 44 to be extremal. Section 4 discusses some results regarding bordered double circulant codes. Section 5 studies the sufficient condition for bordered DC codes to be LCD codes. Section 6 concludes the paper. All computations have been done by MAGMA \cite{ref25}.

\section{Preliminaries}
\subsection{Linear codes}
An $[n,k]$ linear code $\mathcal C$  of length $n$ and dimension $k$ over the Galois field $\mathbb F_q$ is a linear subspace of $\mathbb F_q^n$, where $q$ is a prime power. Any \( k \) linearly independent vectors from this linear subspace form a \( k \times n \) matrix $G$ called the generator matrix of the linear code $\mathcal C$.  All linear combinations of the rows of the generator matrix generate all  codewords in $\mathcal C$. 

The linear code $\mathcal C$ can also be determined by an \((n-k) \times n\) parity check matrix \( H \),  where \( H \) satisfies  $HG^T = \boldsymbol 0$, with 
$\boldsymbol 0$ representing a zero matrix. The dual code $\mathcal C^\perp=\{c^{\prime} \in \mathbb F_q^n \mid cc^{\prime}=0, \forall c \in \mathcal C\}$ of $\mathcal C$ is an $[n, n-k]$ linear code with the generator matrix \( H \) and the parity check matrix \( G \). The code $\mathcal C$ is self-orthogonal provided $\mathcal C \subseteq \mathcal C^\perp$, self-dual provided $\mathcal C = \mathcal C^\perp$, and dual-containing provided $\mathcal C^\perp \subseteq \mathcal C$. 

\subsection{Cyclic codes}
A linear code $\mathcal C$ of length $n$ over $\mathbb F_{q}$ is called a cyclic code if $(c_0, c_1,\cdots ,c_{n-1}) \in \mathcal C $ implies $(c_{n-1}, c_0, \cdots ,c_{n-2}) \in \mathcal C $. The cyclic code $\mathcal C$ is an ideal of a quotient ring $\mathbb F_{q}[x]/(x^n-1)$, which is generated by a monic factor polynomial $g(x)$ of $x^n-1$. Hence codewords in $\mathcal C$ are often represented in the polynomial form, and $g(x)$ is called the generator polynomial of $\mathcal C$. The dual code of $\mathcal C$ is still a cyclic code. Let $h(x)=(x^n-1)/g(x)$, and define $g^\perp(x)=h(0)^{-1}h^*(x)$, where $h^*(x)$ is the reciprocal polynomial of $h(x)$. Then $\mathcal C^\perp$ is generated by $g^\perp(x)$ \cite{ref24}. 

\subsection{Double Circulant codes}
An $m\times m$ circulant matrix over $\mathbb F_q$ is defined as
$$
    \begin{pmatrix}
          a_0&a_1&\cdots &a_{m-1}\\
          a_{m-1}&a_0&\cdots &a_{m-2}\\
          \vdots&\vdots&\ddots&\vdots\\
          a_1&a_2&\cdots &a_0\\
    \end{pmatrix},
$$
where $a_i\in \mathbb F_q,i=0,1,\cdots,m-1$.
The pure double circulant code and the bordered double circulant code are linear codes with generator matrices of the form$$(I_m, A)$$and 
$$
    \begin{pmatrix}
          &  &  &\alpha&1&1&\cdots &1\\
          &  &  &-1& & & & \\
          &I_{m+1}&  &-1& &A& & \\
          &  &  &\vdots& & & & \\
          &  &  &-1& & & & \\
    \end{pmatrix}
$$
respectively, where $\alpha \in \mathbb F_q$, $I_m$ is the $m\times m$ identity matrix and $A$ is an $m\times m$ circulant matrix.

\subsection{Hamming distance and Hamming weight}
The Hamming distance $d(\textbf{\emph{x}}, \textbf{\emph{y}})$ between two codewords $\textbf{\emph{x,y}}\in \mathcal C$ is defined to be the number of coordinates in which $\textbf{\emph{x}}$ and $\textbf{\emph{y}}$ differ. The minimum distance of a code $\mathcal C$ is the smallest distance between distinct codewords and is important in determining the error-correcting capability of $\mathcal C$. The Hamming weight $\mathrm{wt}(\boldsymbol x)$ of a vector $\boldsymbol x\in \mathbb F_q^n$ is the number of nonzero coordinates in $\boldsymbol x$. When $\mathcal C$ is a linear code, the minimum distance $d$ is the same as the minimum weight of the nonzero codewords of $\mathcal C$.

Let $\mathcal C$ be an $[n,k,d]_q$ self-dual code. For the case $q=2$, 
\begin{flalign}\label{eq1}
&&
d(\mathcal C)\leq \left\{\begin{aligned}
4\lfloor {\frac {n}{24}} \rfloor +4,  n\not\equiv22 \, \mathrm{mod}\, 24,\\
4\lfloor {\frac {n}{24}} \rfloor +6,  n\equiv22 \, \mathrm{mod}\, 24.
\end{aligned}\right. 
&&
\end{flalign}
The self-dual code $\mathcal C$ is called extremal if the above equality holds \cite{ref2}.  

\section{Construction of self-dual double circulant codes}

Double circulant codes are a generalization of cyclic codes. Similarly to cyclic codes, double circulant codes can be represented by polynomials. Let $R:=\mathbb F_2[x]/(x^m-1)$. Define a map $\phi:\mathbb F_2^{2m} \to R^2$ by $$\phi(c_{0, 0}, c_{0, 1}, \cdots , c_{0, m-1}, c_{1, 0}, c_{1, 1}, \cdots , c_{1, m-1})=(c_0(x), c_1(x)),$$ where $(c_{0, 0}, c_{0, 1}, \cdots , c_{0, m-1}, c_{1, 0}, c_{1, 1}, \cdots , c_{1, m-1}) \in \mathbb F_2^{2m}$, and $c_0(x)=c_{0, 0}+c_{0, 1}x+\cdots +c_{0, m-1}x^{m-1}, c_1(x)=c_{1, 0}+c_{1, 1}x+\cdots +c_{1, m-1}x^{m-1}$ are polynomials in $R$. It is evident that $\phi$ is a one-to-one correspondence between $\mathbb F_2^{2m}$ and the 2-dimension linear vector space over $R$. Let $\mathcal C$ be a double circulant code of length $n=2m$, then $\mathcal C$ is a subspace of $\mathbb F_2^{2m}$. Hence the generator matrix $G$ of $\mathcal C$ can be expressed by the polynomial matrix as follows:
$$\begin{pmatrix}
          1&f(x)\\
          x&xf(x)\\
          \vdots&\vdots\\
          x^{m-1}&x^{m-1}f(x)
    \end{pmatrix},$$
where the coefficients of \(f(x)\) correspond to the first row of the circulant matrix in the generator matrix of $\mathcal C$. The remaining entries are expressed similarly. 
We know that each codeword in $\mathcal C$ is actually a linear combination of the rows in the generator matrix $G$. In this case we can write $$\mathcal C:=R(1,f(x))=\{(r(x),r(x)f(x))|r(x) \in R\},$$where $r(x)f(x)$ is calculated in $R$. We refer to 
$(1,f(x))$ as the generator of $\mathcal C$. 

\begin{lemma}\cite{ref24}\label{lemma1}
    If $G=[I_k, A]$ is a generator matrix for the $[n, k]_2$ code $\mathcal C$ in standard form, then $H=[-A^T, I_{n-k}]$ is a parity check matrix for $\mathcal C$. 
\end{lemma}

\begin{lemma}\label{lemma99}
    Let $\mathcal C$ be a double circulant code with the generator $(1, f(x))$ over $\mathbb F_2$. $\mathcal C$ is self-dual if and only if $f(x)\overline{f(x)}=1\,  (\mathrm{mod}\,  x^m-1)$, where $\overline{f(x)}=x^mf(\dfrac{1}{x})\, (\mathrm{mod}\, x^m-1)$.
\end{lemma}
\begin{proof}
    Let $f(x)=f_0+f_1x+\cdots +f_{m-1}x^{m-1}$, $f_i\in \mathbb F_2,\ i=0,1,\cdots,m-1$. Then the generator matrix of $\mathcal C$ is
    $$
        \begin{pmatrix}
        1&0&\cdots&0&f_0&f_1&\cdots&f_{m-1}\\
        0&1&\cdots&0&f_{m-1}&f_0&\cdots&f_{m-2}\\
        \vdots&\vdots&\ddots&\vdots&\vdots&\vdots&\ddots&\vdots\\
        0&0&\cdots&1&f_1&f_2&\cdots&f_0
        \end{pmatrix}.
    $$
    According to Lemma \ref{lemma1}, the parity check matrix of $\mathcal C$ is 
    $$
        \begin{pmatrix}
        f_0&f_{m-1}&\cdots&f_1&1&0&\cdots&0\\
        f_{1}&f_0&\cdots&f_{2}&0&1&\cdots&0\\
        \vdots&\vdots&\ddots&\vdots&\vdots&\vdots&\ddots&\vdots\\
        f_{m-1}&f_{m-2}&\cdots&f_0&0&0&\cdots&1
        \end{pmatrix}.
    $$
    Thus the generator of $\mathcal C^\perp$ is $(\overline{f(x)}, 1)$, where $\overline{f(x)}=f_0+f_{m-1}x+\cdots +f_1x^{m-1}=x^mf(\dfrac{1}{x})\, (\mathrm{mod}\, x^m-1)$. If $\mathcal C$ is self-orthogonal, then $(1, f(x))$ is a codeword in $\mathcal C^\perp$, so
    $$
        (1, f(x))=k(x)(\overline{f(x)}, 1)=(k(x)\overline{f(x)}, k(x)),
    $$
    where $k(x)\in R$. Then $k(x)=f(x)\, (\mathrm{mod}\, x^m-1)$ and $f(x)\overline{f(x)}=1\,  (\mathrm{mod}\,  x^m-1)$. Since the dimensions of $\mathcal C$ and $\mathcal C^\perp$ are identical, $\mathcal C$ is not only self-orthogonal, but also self-dual. Therefore, if $\mathcal C$ is self-dual, we have $f(x)\overline{f(x)}=1\,  (\mathrm{mod}\,  x^m-1)$. The reverse is also true. 
\end{proof}
\begin{remark}
     In \cite{ref6}, MacWilliams presented that $AA^T=I_m$ is implied by the corresponding polynomial equation $a(x)a(x)^T=1\,\mathrm{mod}\,(x^m-1)$, where $A$ is an $m\times m$ orthogonal circulant matrix with its first row representing the coefficients of $a(x)$, and $a(x)^T$ denotes $a(x^{m-1})\,\mathrm{mod}\,(x^m-1)$. This statement is equivalent to Lemma \ref{lemma99}. Exactly, we here rewrite and give a full proof.
\end{remark}
\begin{theorem}\label{theorem2}
    Let $f(x)\in R$. If $f(x)\overline{f(x)}=1\,  (\mathrm{mod}\,  x^m-1)$, then $f^*(x)\overline{f^*(x)}=1\,  (\mathrm{mod}\,  x^m-1)$, and $x^if(x)\overline{x^if(x)}=1\,  (\mathrm{mod}\,  x^m-1)$, where $f^*(x)$ is the reciprocal polynomial of $f(x)$ and $i$ denotes a natural number. 
\end{theorem}
\begin{proof}
    The reciprocal polynomial of $f(x)$ is $f^*(x)=x^{deg(f(x))} f(\dfrac{1}{x})$, then
\begin{flalign*}
    &&
\begin{aligned}
    f^*(x)\overline{f^*(x)}&=x^{deg(f(x))} f(\frac{1}{x})x^mx^{-deg(f(x))}f(x)=f(x)x^mf(\frac{1}{x})=1\,  (\mathrm{mod}\,  x^m-1),\\
    x^if(x)\overline{x^if(x)}&=x^if(x)x^mx^{-i}f(\frac{1}{x})=f(x)x^mf(\frac{1}{x})=1\,  (\mathrm{mod}\,  x^m-1).
    \end{aligned}
    &&
\end{flalign*}
    
\end{proof}

Theorem \ref{theorem2} states that if the double circulant code generated by $(1,f(x))$ is self-dual, then the double circulant code generated by $(1,f^*(x))$ or $(1,x^if(x))$ is also self-dual. In fact, the double circulant codes generated by $(1,f(x))$, $(1,f^*(x))$, and $(1,x^if(x))$ are permutation equivalent.

\begin{theorem}
    The double circulant code $\mathcal C_1$ generated by $(1,f(x))$ over $\mathbb F_2$ is equivalent to the double circulant codes $\mathcal C_2$ and $\mathcal C_3$, which are generated by $(1,f^*(x))$ and $(1,x^if(x))$ respectively.
\end{theorem}
\begin{proof}
    Let $G_1$, $G_2$ and $G_3$ denote the generator matrices of $\mathcal C_1$, $\mathcal C_2$ and $\mathcal C_3$ respectively. Define $B_m$ as an $m\times m$ circulant matrix of the following form:$$B_m=\begin{pmatrix}
        0&1&0&\cdots&0\\
        0&0&1&\cdots&0\\
        \vdots&\vdots&\vdots&\ddots&\vdots\\
        0&0&0&\cdots&1\\
        1&0&0&\cdots&0\\
    \end{pmatrix}_{m\times m}$$
    Then the permutation matrix that maps $G_1$ to $G_3$ is given by
    $$ P_1=\begin{pmatrix}
        I_m&\boldsymbol{0}\\
        \boldsymbol{0}&B_m^{i}
    \end{pmatrix},$$
    which is easily verifiable. Let $d\ (d\le m-1)$ represent the degree of $f(x)$. Then 
\setcounter{MaxMatrixCols}{20}
\begin{align*}
    G_1 &= \begin{pmatrix}
     & & & &f_0&f_1&\cdots&f_{d-1}&f_d&\raisebox{-0.4ex}{0}&\cdots&\raisebox{-0.4ex}{0}\\
     & & & &\raisebox{-0.4ex}{0}&f_0&\cdots&f_{d-2}&f_{d-1}&f_d&\cdots&\raisebox{-0.4ex}{0}\\
     & &I_m& & & & &\cdots& & & & \\
     & & & &f_2&f_3&\cdots&0&0&0&\cdots&f_1\\
     & & & &f_1&f_2&\cdots&f_d&0&0&\cdots&f_0\\
    \end{pmatrix}, \\
    G_2 &= \begin{pmatrix}
     & & & &f_d&f_{d-1}&\cdots&f_1&f_0&0&\cdots&0\\
     & & & &\raisebox{-0.4ex}{0}&f_d&\cdots&f_2&f_1&f_0&\cdots&0\\
     & &I_m& & & & &\cdots& & & & \\
     & & & &f_{d-2}&f_{d-3}&\cdots&0&0&0&\cdots&f_{d-1}\\
     & & & &f_{d-1}&f_{d-2}&\cdots&f_0&0&0&\cdots&f_d\\
    \end{pmatrix}.
\end{align*}
Let $D_i$ be the $i\times i$ elementary transformation matrix in the following form:
    $$D_i=\begin{pmatrix}
        & & & &1\\
         & & &1& \\
         & &\begin{sideways}$\ddots$\end{sideways}\qquad& & \\
         &1& & & \\
        1& & & & \\
    \end{pmatrix}_{i\times i}.$$
    Then \setcounter{MaxMatrixCols}{20}
    $$D_mG_1=\begin{pmatrix}
     & & & &f_1&f_2&\cdots&f_d&0&0&\cdots&f_0\\
     & & & &f_2&f_3&\cdots&0&0&0&\cdots&f_1\\
     & &D_m& & & & &\cdots& & & & \\
     & & & &0&f_0&\cdots&f_{d-2}&f_{d-1}&f_d&\cdots&0\\
     & & & &f_0&f_1&\cdots&f_{d-1}&f_d&0&\cdots&0\\
\end{pmatrix}.$$
For the $m\times m$ matrix on the right-hand side of \( D_mG_1 \), the corresponding circulant matrix in \( G_2 \) can be obtained by swapping the \( i \)-th (\( 1 \leq i \leq d \)) column with the \( (d-i+1) \)-th column and the \( j \)-th (\( d+1 \leq j \leq m \)) column with the \( (m+d-j+1) \)-th column. Let $P_2$ be the $2m\times 2m$ elementary transformation matrix in the following form:
    $$P_2=\begin{pmatrix}
       D_m& & \\
        &D_d& \\
         &  &D_{m-d}\\
\end{pmatrix}_{2m\times 2m}.$$
Then $D_mG_1P_2=G_2$.
\end{proof}

If $\mathcal C$ is a self-dual code with the generator $(1, f(x))$ 
 and the generator matrix $G=(I_m, A)$, where $A$ is an $m\times m$ circulant matrix with the first row $f_0, f_1, \cdots, f_{m-1}$, then $GG^T=0$. That is,  $(I_m, A)(I_m, A)^T=0$. From this, we obtain 
 $I_mI_m^T+AA^T=0$, so $AA^T=I_m$. Therefore $f_0^2+f_1^2+\cdots+f_{m-1}^2=1$. Hence if  $\mathcal C$ is self-dual with the generator $(1, f(x))$, then $f(x)$ has odd weight. Notice the weight of $f(x)$ refers to the number of nonzero coefficients. 

\begin{lemma}\cite{ref24}\label{lemma2}
    Let $\mathcal C$ be a linear code with the parity check matrix $H$. If $\boldsymbol{c}\in \mathcal C$, the columns of $H$ corresponding to the nonzero coordinates of $\boldsymbol{c}$ are linearly dependent. Conversely, if a linear dependence relation with nonzero coefficients exists among $\omega$ columns of $H$, then there is a codeword in $\mathcal C$ of weight $\omega$ whose nonzero coordinates correspond to these columns. 
\end{lemma}
Since a self-dual code over $\mathbb F_2$ contains only codewords with even weight, and according to Eq.(\ref{eq1}), the minimum distance for such a code to be extremal is 4 when the code length is less than or equal to 20. Therefore, to obtain an extremal double circulant self-dual code of length up to 20, we only need to demonstrate that there is no codeword of weight 2 and that there exists a codeword of weight 4. Equivalently, we need to prove that its parity check matrix has a set of 4 linearly dependent columns but no set of 2 linearly dependent columns, as stated in Lemma \ref{lemma2}. 

\begin{lemma}\cite{ref24}\label{lemma15}
    The following hold:
    \\(1) If $\boldsymbol{x},\boldsymbol{y}\in F_2^n$, then $$\mathrm{wt}(\boldsymbol{x}+\boldsymbol{y})=\mathrm{wt}(\boldsymbol{x})+\mathrm{wt}(\boldsymbol{y})-2\mathrm{wt}(\boldsymbol{x}\cap \boldsymbol{y}),$$ where $\boldsymbol{x}\cap \boldsymbol{y}$ is the vector in $F_2^n$, which has 1s precisely in those positions where both $\boldsymbol{x}$ and $\boldsymbol{y}$ have 1s.
    \\(2) If $\boldsymbol{x},\boldsymbol{y}\in F_2^n$, then $\boldsymbol{x}\cap \boldsymbol{y}\equiv x\cdot y \ (\mathrm{mod}\  2)$.
\end{lemma}

\begin{theorem}\label{theorem3. 3}
    Let $\mathcal C$ be a self-dual double circulant code over $\mathbb F_2$ of length up to 20 with the generator $(1, f(x))$, where $f(x) \in R$. Then $\mathcal C$ is extremal if and only if $\mathrm{wt}(f(x)) = 3$ or $\mathrm{wt}(f(x)) > 3$ and there exist \(i, j, k\) such that the following equation holds:
    $$f(x) + x^i f(x) \equiv x^j + x^k \, (\mathrm{mod} \, x^m - 1), \quad 1 \le i \le m-1, \; 0 \le j \neq k \le m-1.$$
\end{theorem}
\begin{proof}
In the previous text we have obtained the parity check matrix of $\mathcal C$ $$H=(A^T, I).$$It can also be expressed in the form of 
$$
    \begin{pmatrix}
        [f(x)]^T&[xf(x)]^T&\cdots&[x^{m-1}f(x)]^T&[1]^T&[x]^T&\cdots&[{x^{m-1}}]^T
     \end{pmatrix},
$$
where $[f(x)]^T$ denotes the transpose of the coefficient vector of $f(x)$, and the rest entries are expressed similarly. 

Note that $I$ and $A^T$ are both orthogonal matrices. Then they are invertible, and any set of columns of each of them are linearly independent. Hence we only need to consider linearly dependent columns between them. We discuss three scenarios below: 

(1) $\mathrm{wt}(f(x))=1$, i.e. $f(x)=x^i\ (0\le i \le m-1)$. The first column must be linearly dependent with some column in $I$, then in this case $d=2$. 

(2) $\mathrm{wt}(f(x))=3$. If there are two columns that are linearly dependent in $H$, then $x^if(x)\equiv x^j\,  (\mathrm{mod}\,  x^m-1), 0\le i, j\le m-1$. This implies $f(x)=x^k, 0\le k\le m-1$, which is a contradiction with $\mathrm{wt}(f(x))=3$. Therefore, there is no codeword with weight 2.  Additionally, since $\mathrm{wt}(1, f(x))=4$, there exists at least one codeword of weight 4. 

(3) $\mathrm{wt}(f(x))> 3$. There is also no codeword with weight 2, the proof is similar to that of (2). Therefore, we need to clarify under what circumstances there exist four linearly dependent columns. 

(a) Three columns in the matrix $A^T$ are linearly dependent with one column in the matrix $I$. 

(b) Two columns in the matrix $A^T$ are linearly dependent with two columns in the matrix $I$. 

(c) One column in the matrix $A^T$ is linearly dependent with three columns in the matrix $I$. 

For (c), $\mathrm{wt}(f(x))=3$, which contradicts the assumption that $\mathrm{wt}(f(x))> 3$. For (a), since $\mathrm{wt}(f(x))>3$, $\mathrm{wt}(f(x))$ is odd and $n=2m\le 20$. Therefore, $\mathrm{wt}(f(x))=$5, 7 or 9. 

If $\mathrm{wt}(f(x))=$5 or 9. The weight of the sum of any three arbitrary columns $u$, $v$ and $w$ in $A^T$ is $$\mathrm{wt}(u+v+w)\equiv 1+1+1\equiv 3\ (\mathrm{mod}\ 4)$$ as stated in Lemma \ref{lemma15}. Therefore, when $\mathrm{wt}(f(x))=$5 or 9, (a) is impossible.

If $\mathrm{wt}(f(x))=$7. Then $m=$8, 9 or 10.

1) $m=8$. By observing the structure of the parity check matrix, only (b) can be satisfied in this case.

2) $m=9$. Let $f_{i_0}=f_{i_1}=0,\ 0\le i_0< i_1\le m-1$. Let $i_1-i_0=t$. The inner product of the first row and the second row of \( A \) is
\begin{flalign}\label{eq.(2)}
    &&
    f_0f_8+f_1f_0+\cdots+f_{i_0}f_{i_0-1}+f_{i_0+1}f_{i_0}+\cdots+f_{i_1}f_{i_1-1}+f_{i_1+1}f_{i_1}+\cdots+f_8f_7.
    &&
\end{flalign}
Only $f_{i_0}f_{i_0-1}$, $f_{i_0+1}f_{i_0}$, $f_{i_1}f_{i_1-1}$  and $f_{i_1+1}f_{i_1}$ are equal to zero. If there are no duplicates among these four terms, it can be concluded that  Eq.(\ref{eq.(2)})= 1, which contradicts the fact that 
$A$ is an orthogonal matrix. First, $f_{i_0}f_{i_0-1}$, $f_{i_1}f_{i_1-1}$ and $f_{i_1+1}f_{i_1}$ are definitely three distinct terms. The only possibility is $i_0+1=i_1$. Then $t=1$. In this case, $f_{i_0+1}f_{i_0}=f_{i_1}f_{i_1-1}$. Thus there are three terms in Eq.(\ref{eq.(2)}) equal to zero, so Eq.(\ref{eq.(2)})=0. However, at this point, the inner product of the first row and the third row of \( A \) is
\begin{flalign}\label{eq.(3)}
    &&
    f_0f_7+\cdots+f_{i_0}f_{i_0-2}+\cdots+f_{i_0+2}f_{i_0}+\cdots+f_{i_1}f_{i_1-2}+\cdots+f_{i_1+2}f_{i_1}+\cdots+f_8f_6.
    &&
\end{flalign}
If $t=1$, there are no duplicates among $f_{i_0}f_{i_0-2}$, $f_{i_0+2}f_{i_0}$, $f_{i_1}f_{i_1-2}$ and $f_{i_1+2}f_{i_1}$. Then Eq.(\ref{eq.(3)})=1, which contradicts the fact that 
$A$ is an orthogonal matrix. Therefore, when \( m = 9 \), there does not exist a self-dual code with \( \mathrm{wt}(f(x)) = 7 \).

3) $m=10$, Similarly to the discussion in 2). It can be proven that when \( m = 10 \), there does not exist a self-dual code with \( \mathrm{wt}(f(x)) = 7 \).

Therefore, $\mathcal C$ has the minimum weight 4 if and only if (b) is satisfied, i.e. $\mathcal C$ has the minimum weight 4 if and only if there exist \(i, j, k,l\) such that the following equation holds:
\begin{flalign}\label{eq.(4)}
    &&
    x^if(x)+x^jf(x) \equiv x^k+x^l\,  (\mathrm{mod}\,  x^m-1), 0\le i\neq j , k\neq l\le m-1
    &&
\end{flalign}
is satisfied. Divide both sides of Eq.(\ref{eq.(4)}) by $x^i$ to get the result in the theorem. 
\end{proof}
\begin{example}
    For $m=4$, let $f(x)=x^2+x+1$, which satisfies $f(x)\overline{f(x)}=1\, (\mathrm{mod}\,  x^4-1)$. Since $\mathrm{wt}(f(x))=3$, the double circulant code generated by $(1, x^2+x+1)$ over $\mathbb F_2$ is an $[8, 4, 4]$ extremal self-dual double circulant code according to Theorem \ref{theorem3. 3}. 
\end{example}
\begin{example}
    For $m=6$, let $f(x)=x^4+x^3+x^2+x+1$, which satisfies $f(x)\overline{f(x)}=1\, (\mathrm{mod}\,  x^6-1)$. It is easy to verify $f(x)+xf(x)=x^5+1\, (\mathrm{mod}\,  x^6-1)$. Thus the double circulant code generated by $(1, x^4+x^3+x^2+x+1)$ over $\mathbb F_2$ is an $[12, 6, 4]$ extremal self-dual double circulant code.
\end{example}
\begin{example}
    For $m=9$, let $f(x)=x^6+x^4+x^3+x+1$, which satisfies $f(x)\overline{f(x)}=1\, (\mathrm{mod}\,  x^9-1)$. It is easy to verify $f(x)+x^3f(x)=x^7+x\, (\mathrm{mod}\,  x^9-1)$. Thus the double circulant code generated by $(1, x^6+x^4+x^3+x+1)$ over $\mathbb F_2$ is an $[18, 9, 4]$ extremal self-dual double circulant code. 
\end{example}

By using the method in Theorem \ref{theorem3. 3}, we can obtain all extremal self-dual DC codes over $\mathbb F_2$ with length up to 20. Some of search results are listed in Table  \ref{tab:1}.

\begin{table}[width=.5\linewidth,cols=3,pos=h]
\caption{Extremal self-dual double circulant codes over $\mathbb F_2$ of length up to 20.}\label{tab:1}
\begin{tabular*}{\tblwidth}{@{} LLLL@{} }
\toprule
$\boldsymbol{m}$ & $\boldsymbol{f(x)}$ & $\boldsymbol{\mathcal C}$\\
\midrule
4 & $x^2+x+1$ & [8, 4, 4]\\
        6 & $x^4+x^3+x^2+x+1$ & [12, 6, 4]\\
        8 & $x^4+x^2+1$ & [16, 8, 4]\\
        8 & $x^6+x^5+x^4+x^2+1$ & [16, 8, 4]\\
        8 & $x^6+x^5+x^4+x^3+x^2+x+1$ & [16, 8, 4]\\
        9 & $x^6+x^4+x^3+x+1$ & [18, 9, 4]\\
        10 & $x^9+x^7+x^5+x^4+1$ & [20, 10, 4]\\
        10 & $x^8+x^7+x^6+x^5+x^4+x^3+x^2+x+1$ & [20, 10, 4]\\
\bottomrule
\end{tabular*}
\end{table}
\begin{theorem}\label{theorem3}
    Let $\mathcal C$ be a double circulant code over $\mathbb F_2$ of length $2m\ (2m\le20)$ with generator $(1, f(x))$, where $f(x) \in R$. If $f(x)=1+x^{\frac{m}{4}}+x^{\frac{m}{2}}$, then $\mathcal C$ is an extremal self-dual double circulant code. 
\end{theorem}
\begin{proof}
   $f(x)\overline{f(x)}=(1+x^{\frac{m}{4}}+x^{\frac{m}{2}})(1+x^{\frac{3m}{4}}+x^{\frac{m}{2}})=1+x^{\frac{3m}{4}}+x^{\frac{m}{2}}+x^{\frac{m}{4}}+x^m+x^{\frac{3m}{4}}+x^{\frac{m}{2}}+x^{\frac{5m}{4}}+x^m \equiv 1\,  (\mathrm{mod}\,  x^m-1). $ According to Lemma \ref{lemma99}, $\mathcal C$ is self-dual. Since $\mathrm{wt}(f(x))=3$, then $\mathcal C$ must be an extremal self-dual DC code according to Theorem \ref{theorem3. 3} . 
\end{proof}

\begin{theorem}\label{thoerem3. 5}
    Let $\mathcal C$ be a self-dual double circulant code over $\mathbb F_2$ of length 22 with the generator $(1, f(x))$, where $f(x) \in R$. When $\mathrm{wt}(f(x))=5$, $\mathcal C$ is extremal if and only if $$\mathrm{wt}(f(x)+x^if(x))\neq 2\quad \forall i,\  1\le i\le m-1.$$
\end{theorem}
\begin{proof}
    Note that for $n=22$, the minimum distance for $\mathcal C$ to be extremal is 6. It is clear when $\mathrm{wt}(f(x))=1$ or $\mathrm{wt}(f(x))=3$, the code $\mathcal C$ is not extremal. When $\mathrm{wt}(f(x))=5$, we have $\mathrm{wt}(1, f(x))=6$, indicating that there exist codewords of weight 6. Thus we just need to consider the conditions that there is no codeword of weight 2 and 4. 
    
    Similar to the proof of Theorem \ref{theorem3. 3}, when $\mathrm{wt}(f(x))=5$, there is no codeword of weight 2. Let's consider the case there exists a codeword of weight 4,  that is,  four columns of the parity check matrix are linearly dependent. 
    
    (a) Three columns in the matrix $A^T$ are linearly dependent with one column in the matrix $I$. 

    (b) Two columns in the matrix $A^T$ are linearly dependent with two columns in the matrix $I$. 

    (c) One column in the matrix $A^T$ is linearly dependent with three columns in the matrix $I$. \\
    It is obvious (c) is impossible.  Given that $\mathrm{wt}(f(x))=5$, when we add any two different columns $\textbf{\emph{x}}$ and $\textbf{\emph{y}}$ of $A^T$ together. The weight of $\textbf{\emph{x}}+\textbf{\emph{y}}$ is given by $$\mathrm{wt}(\textbf{\emph{x}}+\textbf{\emph{y}})\equiv 1+1+0\equiv 2\, (\mathrm{mod}\, 4).$$ Similarly, if we add any three different columns $\textbf{\emph{x}}$, $\textbf{\emph{y}}$ and $\bold z$ of $A^T$ together, we can get :$$\mathrm{wt}(\textbf{\emph{x}}+\textbf{\emph{y}}+\bold z)\equiv 2+1+0\equiv 3\, (\mathrm{mod}\, 4).$$ Therefore, (a) is also impossible, the result is now clear. 
\end{proof}

\begin{lemma}\label{lemma3.6}
    Let $\mathcal C$ be a double circulant code generated by $(1, f(x))$ over $\mathbb F_2$. If $\mathrm{wt}(f(x))\equiv 3\, (\mathrm{mod}\, 4)$. Then $\mathcal C$ is doubly even. 
\end{lemma}

\begin{proof}
    Let $\textbf{\emph{x}}$ and $\textbf{\emph{y}}$ be the rows of the generator matrix of $\mathcal C$, then $\mathrm{wt}(\textbf{\emph{x}})\equiv \mathrm{wt}(\textbf{\emph{y}})\equiv 0\, (\mathrm{mod}\, 4)$. Since $\mathcal C$ is self-dual, $\textbf{\emph{x}} \cdot \textbf{\emph{y}}\equiv 0\, (\mathrm{mod}\, 2)$, and  $2\mathrm{wt}(\textbf{\emph{x}}\cap \textbf{\emph{y}})\equiv 0\, (\mathrm{mod}\, 4)$, then 
    $$\mathrm{wt}(\textbf{\emph{x}}+\textbf{\emph{y}})=\mathrm{wt}(\textbf{\emph{x}})+\mathrm{wt}(\textbf{\emph{y}})-2\mathrm{wt}(\textbf{\emph{x}}\cap \textbf{\emph{y}})\equiv 0+0+0=0\, (\mathrm{mod}\, 4).$$Now proceed by induction as every codeword is a sum of rows of the generator matrix. 
\end{proof}

\begin{theorem}\label{theorem3. 7}
    Let $\mathcal C$ be a self-dual double circulant code over $\mathbb F_2$ of length $2m$ \, $(12\le m \le 22)$ with the generator $(1, f(x))$, where $f(x) \in R$. When $\mathrm{wt}(f(x))=7$, $\mathcal C$ is extremal if and only if 
    \begin{align*}\left\{\begin{aligned}
        \mathrm{wt}(f(x)+x^if(x))&\neq 2 \quad\forall i,\  1\le i\le m-1,\\\mathrm{wt}(f(x)+x^if(x)+x^jf(x)) &\neq 1 \quad\forall i, j,\  1\le i\neq j \le m-1
    \end{aligned}\right. \end{align*}
    hold. 
\end{theorem}
\begin{proof}
    The minimum distance for $\mathcal C$ to be extremal is 8 according to Eq.(\ref{eq1}). Since $\mathrm{wt}(f(x))=7\equiv3\,  (\mathrm{mod}\,  4)$, there is no codeword of weight 2 and 6 according to Lemma \ref{lemma3.6}. Thus we only need to consider the case there is no linear correlation of four columns in the parity check matrix of $\mathcal C$. Following the proof of Theorem \ref{thoerem3. 5}, the result is clear, we will not elaborate further here. 
\end{proof}
\begin{table}[width=.5\linewidth,cols=3,pos=h]
\caption{Extremal self-dual double circulant codes over $\mathbb F_2$ of length between 24 and 44.}\label{tab:2}
\begin{tabular*}{\tblwidth}{@{} LLLL@{} }
\toprule
$\boldsymbol{m}$ & $\boldsymbol{f(x)}$ & $\boldsymbol{\mathcal C}$\\
\midrule
12 & $x^8+x^6+x^5+x^4+x^3+x+1$ & [24, 12, 8]\\
        16 & $x^9+x^8+x^7+x^6+x^5+x^3+1$ & [32, 12, 8]\\
        20 & $x^{10}+x^9+x^8+x^4+x^3+x+1$ & [40, 20, 8]\\
\bottomrule
\end{tabular*}
\end{table}
When the code length is between 24 and 44, due to the particularity of $\mathrm{wt}(f(x))=7$, it is convenient for us to discuss the conditions for self-dual DC codes to be extremal. However, for longer code length or $\mathrm{wt}(f(x))>7$. If the methods similar to the previous theorems are applied, the results become cumbersome and may not significantly aid in finding the extremal self-dual DC code. Therefore, only the special cases are discussed in this paper. Table \ref{tab:2} shows some polynomials that satisfy the conditions in Theorem \ref{theorem3. 7} and the parameters of the corresponding codes. 

Since there is a one-to-one correspondence between orthogonal circulant matrices and double circulant codes, once we determine the number of $m \times m$ orthogonal circulant matrices, we can determine the number of $[2m, m]$ self-dual double circulant codes.

Let $\mathbb F_q$ be the Galois filed and $m$ be a positive integer relatively prime to $q$. Let $h=\mathrm{ord}_m(q)$ and let $\alpha$ be a primitive $m$th root of unit in $\mathbb F_{q^h}$. Then for each integer $s$ with $0\le s \le m$, the minimal polynomial of $\alpha^s$ over $\mathbb F_q$ is $M_{\alpha^s}(x)=\prod \limits_{i\in C_s}(x-\alpha^i)$, where $C_s$ is the $q$-cyclotomic coset of $s$ modulo $m$. Furthermore, $$x^m-1=\prod \limits_s M_{\alpha^s}(x)$$ is the factorization of $x^m-1$ into irreducible factors over $\mathbb F_q$, where $s$ runs through a set of representations of the $q$-cyclotomic coset of $s$ modulo $m$. Given that both $C_s$ and $C_{-s}$ are $q$-cyclotomic coset modulo $m$, the reciprocal polynomial of each irreducible factor of $x^m - 1$ is also an irreducible factor of $x^m-1$. We decompose $x^m - 1$ in an alternative form as follows:
$$x^m-1=(x-1)\prod_{i=1}^rf_i(x),$$
where \( f_i(x) \) is the minimal polynomial of some $\alpha^s$. Let $k_i$ be the degree of $f_i(x)$. By arranging the $f_i(x)s$ so that $f_1(x),f_2(x), \cdots ,f_t(x)$ are self-reciprocal, thus their degrees are even, i.e.  $k_i=2c_i,\,1\le i \le t$ ($f_0(x)=x-1$ is a special case) and arranging the other $f_i(x)s$ in pairs so that $$g_j(x)=f_j(x)f_j^*(x),$$ $g_j(x)$ is of degree $2d_j$. Hence
$$x^m-1=(x-1)\prod_{i=1}^tf_i(x)\prod_{j=t+1}^ug_j(x),$$ where $u=(r+t)/2$.

MacWilliams \cite{ref6} presented the number of $m \times m$ orthogonal circulant matrices over $\mathbb F_2$ in 1971 as the following lemma shows. 
\begin{lemma} \cite{ref6}\label{lemma3}
    For $q=2$ and $\mathrm{gcd}(m, q)=1$, the number of $m \times m$ orthogonal circulant matrices over $\mathbb F_2$ is denoted by $|O_m|$, and
    $$|O_m| =\prod _{i=1}^t(2^{c_i}+1)\prod _{j=t+1}^u(2^{d_j}-1).$$
\end{lemma}
    
\begin{lemma} \cite{ref6}
    For $q=2$ and m=sq. If s is odd, $|O_{2s}|=2^{(s+1)/2}|O_s|$. If s is even but $s/2$ is odd, $|O_{2s}|=2^{s/2+1}|O_s|$. If $s\equiv 0\, (\mathrm{mod}\, 4), \,|O_{2s}|=2^{s/2}|O_s|$.
\end{lemma}
\begin{theorem}\label{theorem4}
    For any positive integer $m\ (m\ge 2)$. If the number of $m\times m$ orthogonal circulant matrices is m, then the self-dual double circulant code $\mathcal C$ over $\mathbb F_2$ of length $2m$ is not extremal. 
\end{theorem}
\begin{proof}
     Notice that in polynomial residue ring $R=\mathbb F_2[x]/(x^m-1)$. There are at least $m$ polynomials $f(x)$ satisfy $f(x)\overline{f(x)}\equiv 1\,  (\mathrm{mod}\,  x^m-1)$, they are $1, x, \cdots , x^{m-1}$, for $x^i\overline{x^i}=x^ix^{m-i}=x^m=1\,  (\mathrm{mod}\,  x^m-1),\,i=0,1,\cdots ,m-1$. Thus if the number of $m\times m$ orthogonal circulant matrices is exactly $m$, then the self-dual double circulant code $\mathcal C$ of length $2m$ is generated by $(1, x^i),\, 0\le i\le m-1$, which means the minimum distance of $\mathcal C$ is obvious 2. In this case there is no extremal self-dual double circulant code of length $2m$ according to Eq.(\ref{eq1}). 
\end{proof}
\begin{example}
    For $m=5,\, x^5-1=(x+1)(x^4+x^3+x^2+x+1)$. By applying Lemma \ref{lemma3}, $t=1, c_1=4/2=2$, so $| O_{5}|=2^2+1=5$. Thus there is no extremal self-dual double circulant codes of length $10$ over $\mathbb F_2$ according to Theorem \ref{theorem4}. 
\end{example}
\begin{example}
    For $m=7,\, x^7-1=(x+1)(x^3+x+1)(x^3+x^2+1)$. By applying Lemma \ref{lemma3}, $u=1, t=0, d_1=3$, so $| O_{7}|=2^3-1=7$. Thus there is no extremal self-dual double circulant codes of length $14$ over $\mathbb F_2$ according to Theorem \ref{theorem4}. 
\end{example}

\section{Bordered Double Circulant Self-Dual Codes}
The generator matrix of a bordered double circulant code $\mathcal C$ over $\mathbb F_2$ is in the form of 
    $$G'=
    \begin{pmatrix}
          &  &  &\alpha&1&1&\cdots &1\\
          &  &  &1& & & & \\
          &I_{m+1}&  &1& &A& & \\
          &  &  &\vdots& & & & \\
          &  &  &1& & & & \\
    \end{pmatrix},$$
where $\alpha \in \mathbb F_2$, let
    $$A=
    \begin{pmatrix}
        f_0&f_1&\cdots&f_{m-1}\\
        f_{m-1}&f_0&\cdots&f_{m-2}\\
        \vdots&\vdots&\ddots&\vdots\\
        f_1&f_2&\cdots&f_0
    \end{pmatrix},\ A'=
    \begin{pmatrix}
        \alpha&1&1&\cdots &1\\
        1&f_0&f_1&\cdots&f_{m-1}\\
        1&f_{m-1}&f_0&\cdots&f_{m-2}\\
        \vdots&\vdots&\vdots&\ddots&\vdots\\
        1&f_1&f_2&\cdots&f_0
    \end{pmatrix}.$$
Then $G'=(I_{m+1}, A')$. The bordered DC code is self-dual if and only if $A'$ is an orthogonal matrix. Let $f(x)=\sum_{i=0}^{m-1}f_ix^i$, then we call $A$ the matrix corresponding to $f(x)$, and we refer to $\mathcal C$ as the bordered DC code corresponding to $f(x)$. Certain elementary facts about self-dual bordered  double circulant codes are gathered in the following theorem. 
\begin{theorem}\label{theorem4. 1}
Let $\mathcal C$ be a $[2m+2, m+1]$ self-dual bordered double circulant code  corresponding to $f(x)$ over $\mathbb F_2$. Then:
\\(1) $\alpha=0$;
\\(2) m is odd;
\\(3) $\mathrm{wt}(f(x))$ is even;
\\(4) $\mathcal C$ is self-dual if and only if $AA^T=
    \begin{pmatrix}
        0&1&1&\cdots&1\\
        1&0&1&\cdots&1\\
        1&1&0&\cdots&1\\
        \vdots&\vdots&\vdots&\ddots&\vdots\\
        1&1&1&\cdots&0
    \end{pmatrix}$;
\\(5) $\mathcal C$ is self-dual if and only if $f(x)\overline{f(x)}=x+x^2+\cdots+x^{m-1}\,  (\mathrm{mod}\,  x^m-1).$
    
\end{theorem}
\begin{proof}
    If $\alpha=1$. Since $A'$ is an orthogonal matrix, the inner product of the first row of $A'$ with any other row of $A'$ is equal to 0, which implies that $\mathrm{wt}(f(x))$ is odd. However, the inner product of the second row of $A'$ with itself is equal to 1, indicating that $\mathrm{wt}(f(x))$ is even, which is a contradiction. Therefore, if the bordered double circulant code is self-dual. Then $\alpha$ must be 0. 
    
    Since the inner product of the first row of $A'$ with itself is 1, so $m$ is odd. The inner product of the first row  and the second row of $A'$ is 0, so $\mathrm{wt}(f(x))$ is even, proving (2) and (3). 
    
    For (4). Since $\mathcal C$ is self-dual, thus $A'A'^T=I_{m+1}$, yielding $$
    \begin{pmatrix}
        0&\bold 1\\
        \bold 1^T&A
    \end{pmatrix}
    \begin{pmatrix}
        0&\bold 1\\
        \bold 1^T&A^T
    \end{pmatrix}
    =I_{m+1},
    $$
    where $\bold 1$ is a vector of length $m$ with each element equal to 1. Therefore, 
    $$\begin{pmatrix}
        1&\boldsymbol{{f_0+f_1+\cdots +f_{m-1}}}\\
        (\boldsymbol{{f_0+f_1+\cdots +f_{m-1}}})^T&E_m+AA^T
    \end{pmatrix}=I_{m+1},$$
    where $\boldsymbol{{f_0+f_1+\cdots +f_{m-1}}}$ is a vector of length $m$ with each element equal to $f_0+f_1+\cdots +f_{m-1}$ and $E_m$ is an $m\times m$ all-ones matrix. Then $E_m+AA^T=I_m$ implies (4), and the inverse is obvious. 
    
    For (5), it suffices to prove the equivalence of $AA^T=E_m-I_m$ and $f(x)\overline{f(x)}=x+x^2+\cdots+x^{m-1}$.
    Let $$A=
    \begin{pmatrix}
        f_0&f_1&\cdots&f_{m-1}\\
        f_{m-1}&f_0&\cdots&f_{m-2}\\
        \vdots&\vdots&\ddots&\vdots\\
        f_1&f_2&\cdots&f_0
    \end{pmatrix},A^T=
    \begin{pmatrix}
        f_0&f_{m-1}&\cdots&f_1\\
        f_1&f_0&\cdots&f_2\\
        \vdots&\vdots&\ddots&\vdots\\
        f_{m-1}&f_{m-2}&\cdots&f_0
    \end{pmatrix}.$$
    Let $$B=\begin{pmatrix}
        0&1&0&\cdots&0\\
        0&0&1&\cdots&0\\
        \vdots&\vdots&\vdots&\ddots&\vdots\\
        0&0&0&\cdots&1\\
        1&0&0&\cdots&0\\
    \end{pmatrix}_{m\times m},B^2=\begin{pmatrix}
        0&0&1&0&\cdots&0\\
        0&0&0&1&\cdots&0\\
        \vdots&\vdots&\vdots&\vdots&\ddots&\vdots\\
        0&0&0&0&\cdots&1\\
        1&0&0&0&\cdots&0\\
        0&1&0&0&\cdots&0\\
    \end{pmatrix},\cdots,B^{m-1}=\begin{pmatrix}
        0&0&\cdots&0&1\\
        1&0&\cdots&0&0\\
        0&1&\cdots&0&0\\
        \vdots&\vdots&\ddots&\vdots&\vdots\\
        0&0&\cdots&1&0\\
        \end{pmatrix}.$$
       Then $A=f
        _0I_m+f_1B+f_2B^2+\cdots+f_{m-1}B^{m-1}=f(B),\,A^T=f
        _0I_m+f_{m-1}B+f_{m-2}B^2+\cdots+f_1B^{m-1}=\overline{f(B)}$. According to (4), $\mathcal C$ is self-dual if and only if $AA^T=B+B^2+\cdots+B^{m-1}$. 
        Since $B^m = I_m$, there naturally exists a one-to-one mapping between the circulant matrix $B$ and $x$ in $R$. Then $AA^T=B+B^2+\cdots+B^{m-1}$ equivalent to $f(x)\overline{f(x)}=x+x^2+\cdots+x^{m-1}\,  (\mathrm{mod}\,  x^m-1)$, which completes the proof.
\end{proof}

Similar to the discussion of double circulant codes, we consider the self-duality of bordered double circulant codes. 
\begin{theorem}
    If a bordered DC code with generator matrix corresponding to $f(x)$ over $\mathbb F_2$ is self-dual, then the bordered DC code corresponding to $f^*(x)$ is also self-dual, where $f^*(x)$ is the reciprocal polynomial of $f(x)$.
\end{theorem}
\begin{proof}
    The conclusion can be directly obtained from  $$f^*(x)\overline{f^*(x)}=x^{deg(f(x))} f(\frac{1}{x})x^mx^{-deg(f(x))}f(x)=f(x)x^mf(\frac{1}{x})=f(x)\overline{f(x)}=x+x^2+\cdots+x^{m-1}\,  (\mathrm{mod}\,  x^m-1),$$ according to Theorem \ref{theorem4. 1} (5).
\end{proof}

\begin{theorem}
    Let $\mathcal C$ be a $[2m+2, m+1]$  bordered double circulant code over $\mathbb F_2$ with the generator matrix corresponding to $x+x^2+\cdots +x^{m-1}$. Then $\mathcal C$ is self-dual. Moreover, $\mathcal C$ is extremal when $m\le 9$, but not extremal when $m> 9$.
\end{theorem}
\begin{proof}
    Let $f(x)=x+x^2+\cdots +x^{m-1}$. It is easily to verify that $f(x)\overline{f(x)}=x+x^2+\cdots+x^{m-1}\,  (\mathrm{mod}\,  x^m-1)$. Thus $\mathcal C$ is self-dual according to Theorem \ref{theorem4. 1}. When $m\le 9$, the minimum distance for $\mathcal C$ to be extremal is 4. Since $\mathcal C$ contains only even weight codewords, it is sufficient to prove there is no codeword of weight 2, and there is at least one codeword of weight 4 in $\mathcal C$. According to Lemma \ref{lemma1}, the parity check matrix of $\mathcal C$ is given by 
    $$H=\begin{pmatrix}
        A'^T&I_{m+1}
    \end{pmatrix}
    =
    \begin{pmatrix}
        A'&I_{m+1}
    \end{pmatrix}.$$
    Since both $A'$ and $I_{m+1}$ are orthogonal matrices. We only need to consider linearly dependent columns between them. It is obvious there is no codeword of weight 2 by observing the structure of the two matrices. Furthermore, adding any two rows of the generator matrix together results in a codeword weighing 4, which means $\mathcal C$ is extremal when length less than 20. However, for $m> 9$, the minimum distance for $\mathcal C$ to be extremal is greater than 4. Then the results are clear.
\end{proof}

\begin{definition}
    Let $f(x)=f_0+f_1x+\cdots+f_{m-1}x^{m-1}\in R$. Define the complement polynomial of \( f(x) \) as $$\widehat{f}(x)=(f_0+1)+(f_1+1)x+\cdots+(f_{m-1}+1)x^{m-1}.$$
\end{definition}

\begin{theorem}\label{theorem4.5}
    Let $f(x)=f_0+f_1x+\cdots+f_{m-1}x^{m-1}\in R$. When $m$ is odd. If $f(x)\overline{f(x)}=1\ (\mathrm{mod}\ x^m-1)$. Then $\widehat{f}(x)\overline{\widehat{f}(x)}=x+x^2+\cdots+x^{m-1}\ (\mathrm{mod}\ x^m-1)$.
\end{theorem}
\begin{proof}
    From $f(x)\overline{f(x)}=1\ (\mathrm{mod}\ x^m-1)$, we can deduce that:
\begin{flalign*}
    &&
    \left\{
\begin{aligned}
    &f_0^2 + f_1^2 + \cdots + f_{m-1}^2 = 1, \\
    &f_0f_{m-1} + f_1f_0 + \cdots + f_{m-1}f_{m-2} = 0, \\[6pt]
    &\makebox[10em]{\hfill\text{$\cdots$}\hfill} \\[6pt]
    &f_0f_2 + f_1f_3 + \cdots + f_{m-1}f_1 = 0, \\
    &f_0f_1 + f_1f_2 + \cdots + f_{m-1}f_0 = 0.
\end{aligned}
\right.
    &&
\end{flalign*}
    Since $\widehat{f}(x)=(f_0+1)+(f_1+1)x+\cdots+(f_{m-1}+1)x^{m-1}$, $\overline{\widehat{f}(x)}=(f_0+1)+(f_{m-1}+1)x+\cdots+(f_1+1)x^{m-1}$ and $m$ is odd. Therefore, 
    
\begin{flalign}\label{eq.5}
    &&
    \left\{
\begin{aligned}
    \displaystyle \sum_{i=0}^{m-1} (f_i + 1)^2 
    &= \sum_{i=0}^{m-1} f_i^2 + m = 0, \\[10pt]
    \displaystyle \sum_{i=0}^{m-1} (f_i+1)(f_{i-1}+1) 
    &= \sum_{i=0}^{m-1} f_i f_{i-1} + 2 \sum_{i=0}^{m-1} f_i + m = 1, \\[10pt]
    &\cdots \\[10pt]
    \displaystyle \sum_{i=0}^{m-1} (f_i+1) (f_{i+2}+1) &=\sum_{i=0}^{m-1} f_i f_{i+2}+2 \sum_{i=0}^{m-1} f_i+m= 1, \\[8pt]
    \displaystyle \sum_{i=0}^{m-1} (f_i+1) (f_{i+1}+1) &=\sum_{i=0}^{m-1} f_i f_{i+1}+ 2 \sum_{i=0}^{m-1} f_i+m= 1.
\end{aligned}\right.
    &&
    \end{flalign}
Eq.(\ref{eq.5}) indicates that $\widehat{f}(x)\overline{\widehat{f}(x)}=x+x^2+\cdots+x^{m-1}\ (\mathrm{mod}\ x^m-1)$. The result is clear now.
\end{proof}

\begin{corollary}
    When m is odd. There is a one-to-one correspondence between $[2m,m]_2$ self-dual double circulant codes and $[2m+2,m+1]_2$ self-dual bordered double circulant codes.
\end{corollary}
\begin{proof}
    From Theorem \ref{theorem4.5}, if \( A \) is the circulant matrix corresponding to \( f(x) \) and \( f(x)\overline{f(x)} = 1 \ (\mathrm{mod}\ x^m-1)\). Then the double circulant code with the generator matrix \( (I, A) \) is self-dual. Let \( \widehat{A} \) be the circulant matrix corresponding to \( \widehat{f}(x) \). Then the  bordered double circulant code with \( B\) as the generator matrix is self-dual, where $$B=\begin{pmatrix}
          &  &  &0&1&1&\cdots &1\\
          &  &  &1& & & & \\
          &I_{m+1}&  &1& &\widehat{A}& & \\
          &  &  &\vdots& & & & \\
          &  &  &1& & & & \\
    \end{pmatrix}.$$
\end{proof}

\begin{theorem}
    Let $\mathcal C$ be the self-dual double circulant code corresponding to $f(x)$ for the length \( n \leq 18 \) over $\mathbb F_2$. If $\mathcal C$ is extremal. Then the self-dual bordered double circulant code $\widehat{\mathcal C}$ corresponding to \( \widehat{f}(x) \) is also extremal. 
\end{theorem}

\begin{proof}
    From Theorem \ref{theorem3. 3}, a double circulant code over \( \mathbb{F}_2 \) with  length \( n \leq 18 \) is extremal if and only if $\mathrm{wt}(f(x)) = 3$ or $\mathrm{wt}(f(x))> 3$ and there exist \(i, j, k\) such that $f(x)+x^if(x) \equiv x^j+x^k\,  (\mathrm{mod}\,  x^m-1), 1\le i \le m-1, 0\le j\neq k \le m-1$. We only need to consider the latter case, i.e. $\mathrm{wt}(f(x)+x^if(x))=2$, as Theorem \ref{theorem4.5} implies that a self-dual bordered double circulant code can only be obtained when \( \mathrm{wt}(f(x)) \) is even. The parity check matrix of $\widehat{\mathcal C}$ is $$\widehat{H}=\begin{pmatrix}
          0&1&1&\cdots &1&  &  &\\
          1& & & & &  &  &\\
          1&[\widehat{f}(x)]^T&[x\widehat{f}(x)]^T&\cdots&[x^{m-1}\widehat{f}(x)]^T& &I_{m+1}& \\
          \vdots& & & & &  &  &\\
          1& & & & &  &  &\\
    \end{pmatrix},$$
    where $[\widehat{f}(x)]^T$ denotes the transpose of the coefficient vector of $\widehat{f}(x)$, and the rest entries are expressed similarly. From the structure of the parity check matrix, it is evident that no two columns are linearly dependent, meaning there are no codewords of weight 2. Moreover, since $\mathrm{wt}(f(x) + x^i f(x)) = \mathrm{wt}(\widehat{f}(x) + x^i \widehat{f}(x))=2,\  1\le i \le m-1$. When \( \mathrm{wt}(f(x)) > 3 \), it holds that  there must exist four linearly dependent columns in the parity check matrix of $\widehat{\mathcal C}$, which implies the existence of a codeword with weight 4. Therefore, $\widehat{\mathcal C}$ is extremal of length up to 20, completing the proof.
\end{proof}

\section{Double Circulant Complementary Dual Codes}
Linear complementary dual codes intersect trivially with their dual, i.e. $\mathcal C$ is an LCD code if and only if $\mathcal C\cap \mathcal C^{\bot} =\{0\}$, where $\bot$ represents Euclidean dual in this paper. In \cite{ref23}, Guan et al.  studied one-generator quasi-cyclic (QC) codes with Euclidean, Hermitian and symplectic complementary duals from codeword level. Since DC codes are a special case of one-generator QC codes, the conclusions in \cite{ref23} can be directly applied to DC codes. Building on this, the paper also provides the conditions for bordered DC codes to be LCD under the Euclidean inner product. Let’s start by introducing a few lemmas. Note the Euclidean inner product of $\textbf{\emph{x}}=(x_0, x_1, \cdots,x_{n-1}),  \textbf{\emph{y}}=(y_0, y_1, \cdots,y_{n-1})\in \mathbb F_2^n$ is defined as
$$\langle \textbf{\emph{x}}, \textbf{\emph{y}}\rangle=\sum_{i=0}^{n-1}x_iy_i.$$
\begin{lemma}\cite{ref23}
    Let $\mathcal C$ be a linear code over $\mathbb F_2$, then $\mathcal C$ is an LCD code if and only if\, $\forall c_1\in\mathcal C\backslash\{\bold 0\}, \exists c_2\in \mathcal C, \langle c_1, c_2 \rangle \neq0$ holds. 
\end{lemma}
\begin{lemma}\label{Theorem5. 1}
    Let $\mathcal C$ be a double circulant code generated by $(1, f(x))$ over $\mathbb F_2$, then the sufficient and necessary condition for $\mathcal C$ to be an Euclidean LCD code is $$\mathrm{gcd}(1+f(x)\overline{f(x)}, x^m-1)=1.$$
\end{lemma}
    
\begin{proof}
    The statement can be directly proven by Theorem 4.4 in \cite{ref23}.
\end{proof}
\begin{lemma}\label{lemma5.3}
    Let $f(x)\in R$. If $\mathrm{gcd}(f(x)\overline{f(x)}+1,x^m-1)=1$, then $\mathrm{wt}(f(x))$ is even.
\end{lemma}
\begin{proof}
    According to Lemma \ref{Theorem5. 1}, the DC code $\mathcal C$ generated by $(1,f(x))$ is LCD. If $\mathrm{wt}(f(x))$ is odd, all codewords in $\mathcal C$ have even weight. Then the vector obtained by summing all rows of the generator matrix of $\mathcal C$ is $(1,1,\cdots,1)$ which is orthogonal to each codeword in $\mathcal C$. This contradicts to the fact that $\mathcal C$ is LCD. Therefore, $\mathrm{wt}(f(x))$ is even.
\end{proof}
Due to the structural connection between DC codes and bordered DC codes, the conditions for a DC code to be LCD can be used to derive related conclusions for bordered DC codes.
\begin{lemma}\label{lemma5.5}
    Let $\mathcal C$ be a $[2m+2,m+1]$ bordered double circulant code associated with $f(x)$, where $\alpha=0$ in its generator matrix. If $\mathcal C$ is LCD,  then the following hold:
    \\(1) m is odd if and only if $\mathrm{wt}(f(x))$ is odd. 
    \\(2) m is even if and only if $\mathrm{wt}(f(x))$ is even. 
\end{lemma}

\begin{proof}
    When $m$ is odd, if $\mathrm{wt}(f(x))$ is even, then the vector obtained by summing the last $m$ rows of the generator matrix is precisely the first row of the parity check matrix. Actually, this means that $\mathcal C\cap \mathcal C^{\perp} \neq \{0\}$. Therefore, if $m$ is odd, $\mathrm{wt}(f(x))$ is also odd. 
    
    When $\mathrm{wt}(f(x))$ is odd, if $m$ is even, then the vector obtained by summing the last $m$ rows of the generator matrix is precisely the vector obtained by summing the last $m$ rows of the parity check matrix, which means that $\mathcal C\cap \mathcal C^{\perp} \neq \{0\}$. Therefore, if $\mathrm{wt}(f(x))$ is odd, $m$ is also odd.
    Similarly, we can also conclude that $m$ is even if and only if $\mathrm{wt}(f(x))$ is even. 
\end{proof}

\begin{theorem}
    Let $f(x)\in R$. If $\mathrm{gcd}(f(x)\overline{f(x)}+1,x^m-1)=1$, then the bordered DC code associated with  $\widehat{f}(x)$, where $\alpha=0$ in its generator matrix over $\mathbb F_{2}$ is LCD.
\end{theorem}
\begin{proof}
    Let $G=(I_m,A)$ be the generator matrix of the DC code $\mathcal C$, where $A$ is the circulant matrix corresponds to $f(x)$. According to Lemma \ref{Theorem5. 1}, $\mathcal C$ is LCD. Let $$G_b=
    \begin{pmatrix}
          &  &  &0&1&1&\cdots &1\\
          &  &  &1&\widehat{f}_0&\widehat{f}_1&\cdots &\widehat{f}_{m-1} \\
          &I_{m+1}&  &1&\widehat{f}_{m-1}&\widehat{f}_0&\cdots&\widehat{f}_{m-2}\\
          & &  &\vdots&\vdots&\vdots&\ddots& \\
          &  &  &1&\widehat{f}_1&\widehat{f}_2&\cdots &\widehat{f}_0\\
    \end{pmatrix}$$ be the generator matrix of the bordered DC code $\mathcal C_b$ which corresponds to $\widehat{f}(x)$, where $\widehat{f_i}=f_i+1$. Perform elementary row and column operations on \( G_b \): add the first row to each of the remaining rows, move the \((m+2)\)-th column to the second column, and shift the remaining columns one position to the right in sequence resulting in the matrix in the following form:$$G_b'=
    \begin{pmatrix}
          1&0&0&\cdots&0&1&1&\cdots&1\\
          1&1& & & &f_0&f_1&\cdots &f_{m-1}\\
          1&1& &I_m& &f_{m-1}&f_0&\cdots&f_{m-2}\\
          \vdots&\vdots& & & &\vdots&\vdots&\ddots&\vdots \\
          1&1& & & &f_1&f_2&\cdots &f_0\\
    \end{pmatrix}.$$
Let the code generated by \( G_b' \) be denoted as \( \mathcal C_b' \). Clearly, \( \mathcal C_b \) is equivalent to \( \mathcal C_b' \). Since \( \mathcal C \) is an LCD code, the codewords generated by all linear combinations of the last \( m-1 \) rows of \( G_b' \) also form an LCD code. Therefore, it suffices to consider whether, for a codeword \( c\in \mathcal C_b' \) formed by a linear combination of the first row of \( G_b' \) and the remaining rows, there exists a \( c'\in \mathcal C_b' \) such that the inner product of \( c \) and \( c' \) is nonzero. Since \( \mathcal C \) is an LCD code. By Lemma \ref{lemma5.3}, \( \mathrm{wt}(f(x)) \) is even. Therefore, we only need to consider the following four cases:  
    
    (1) \( m \) is odd, and \( c_1 \), the codeword formed by a linear combination of the first row of \( G_b' \) and an odd number of rows from the last \( m-1 \) rows of \( G_b' \), is considered. $c_1=(0,1,\boldsymbol{u},\boldsymbol{v})$, where $\boldsymbol{u},\boldsymbol{v}\in F_2^m$ and both $\mathrm{wt}(\boldsymbol{u})$ and $\mathrm{wt}(\boldsymbol{v})$ are odd. In this case, the inner product of $c_1$ and the first row of $G_b'$ is 1. Therefore, \( c_1 \) is not in the dual of \( \mathcal C_b' \).
    
    (2) \( m \) is odd, and \( c_2 \), the codeword formed by a linear combination of the first row of \( G_b' \) and an even number of rows from the last \( m-1 \) rows of \( G_b' \), is considered. $c_2=(1,0,\boldsymbol{u},\boldsymbol{v})$, where $\mathrm{wt}(\boldsymbol{u})$ is even, $\mathrm{wt}(\boldsymbol{v})$ is odd. In this case, The inner product of \( c_2 \) with the vector obtained by summing all the rows of \( G_b' \) is 1. Therefore, \( c_2 \) is not in the dual of \( \mathcal C_b' \).
    
    (3) \( m \) is even, and \( c_3 \), the codeword formed by a linear combination of the first row of \( G_b' \) and an odd number of rows from the last \( m-1 \) rows of \( G_b' \), is considered. $c_3=(0,1,\boldsymbol{u},\boldsymbol{v})$, where $\mathrm{wt}(\boldsymbol{u})$ is odd, $\mathrm{wt}(\boldsymbol{v})$ is even. In this case, The inner product of \( c_3 \) with the vector obtained by summing all the rows of \( G_b' \) is 1. Therefore, \( c_3 \) is not in the dual of \( \mathcal C_b' \).
    
    (4) \( m \) is even, and \( c_4 \), the codeword formed by a linear combination of the first row of \( G_b' \) and an even number of rows from the last \( m-1 \) rows of \( G_b' \), is considered. $c_4=(1,0,\boldsymbol{u},\boldsymbol{v})$, where both $\mathrm{wt}(\boldsymbol{u})$ and $\mathrm{wt}(\boldsymbol{v})$ are even. In this case, the inner product of $c_4$ and the first row of $G_b'$ is 1. Therefore, \( c_4 \) is not in the dual of \( \mathcal C_b' \).
    
    In conclusion, \( \mathcal C_b' \) is an LCD code, and since \( \mathcal C_b \) is equivalent to \( \mathcal C_b' \), \( \mathcal C_b \) is also an LCD code.
\end{proof}

\begin{theorem}
    Let \( \mathcal C_f \) be the \([2m+2, m+1]\) bordered DC code associated with \( f(x) \) over $\mathbb F_2$, where \(\alpha = 1\) in the generator matrix. If \( m \) is even, \( \mathcal C \) cannot be an LCD code. If \( m \) is odd and $\mathrm{gcd}(f(x)\overline{f(x)}+1,x^m-1)=1$, then \( \mathcal C_f \) and \( \mathcal C_{\widehat{f}} \) are both LCD codes.
\end{theorem}
\begin{proof}
    When \(\alpha = 1\) and \(m\) is even, if \(\mathrm{wt}(f(x))\) is odd, then the codeword generated by the linear combination of all rows of \(G\) is $$(\underbrace{1,1,\cdots,1}_{m+2},\underbrace{0,0,\cdots,0}_{m}).$$This is identical to the first row of the parity check matrix, indicating that \(\mathcal C \cap \mathcal C^\perp \neq \{0\}\).
    If \(\mathrm{wt}(f(x))\) is even, then all rows of the generator matrix of \( \mathcal C \) have even weight. According to Lemma \ref{lemma15}, all codewords in \( \mathcal C \) also have even weight. By summing all rows of the generator matrix, we obtain the codeword $$(\underbrace{1, 1,\cdots, 1}_{2m+2}).$$This codeword is orthogonal to all codewords in \( \mathcal C \), indicating that \( \mathcal C \cap \mathcal C^\perp \neq \{0\} \). Therefore, when \( m \) is even, \( \mathcal C \) cannot be an LCD code.

    When \(\alpha = 1\) and \(m\) is odd, the generator matrix of the bordered DC code $\mathcal C_f$ associated with \(f(x)\) has the following form:  
$$G_f=
    \begin{pmatrix}
          1&  &  & & &1&1&1&\cdots &1\\
          &1& & & &1&f_0&f_1&\cdots &f_{m-1} \\
          & &1& & &1&f_{m-1}&f_0&\cdots&f_{m-2}\\
          & &  &\ddots& &\vdots&\vdots&\vdots&\ddots&\vdots \\
          &  &  & &1&1&f_1&f_2&\cdots&f_0\\
    \end{pmatrix}.$$We will discuss the two cases where \(\mathrm{wt}(f(x))\) is odd and \(\mathrm{wt}(f(x))\) is even.

    (1) When \(\mathrm{wt}(f(x))\) is odd, we first consider all codewords generated by the last \(m-1\) rows of \(G_f\). 
    
    The codewords formed by an odd number of these rows have the form \((0, \boldsymbol{u}, 1, \boldsymbol{v})\), where both \(\mathrm{wt}(\boldsymbol{u})\) and \(\mathrm{wt}(\boldsymbol{v})\) are odd. The codeword generated by summing the last \(m-1\) rows of \(G\) is $$(0,\underbrace{1, 1, \cdots, 1}_{2m+1}).$$ The inner product of this codeword with \((0,\boldsymbol{u}, 1, \boldsymbol{v})\) is 1. Thus, any codeword formed by an odd combination of these rows does not belong to \(\mathcal C_f^{\perp}\). 
    
    For codewords formed by an even linear combination of these rows, the resulting codeword has the form \((0, \boldsymbol{u}, 0, \boldsymbol{v})\). Since $\mathrm{gcd}(f(x)\overline{f(x)}+1,x^m-1)=1$, according to Lemma \ref{Theorem5. 1}, there must exist a codeword \((\beta, \boldsymbol{u}', \gamma, \boldsymbol{v}') \in \mathcal C_f\) such that the inner product of \((0, \boldsymbol{u}, 0, \boldsymbol{v})\) with \((\beta, \boldsymbol{u}', \gamma, \boldsymbol{v}')\) equals 1, where $\beta,\gamma \in \{0,1\}$. Therefore, any codeword formed by an even combination of these rows also does not belong to \(\mathcal C_f^{\perp}\).
    
    Next, we consider all codewords formed by a linear combination of the first row and the remaining rows, the resulting codewords have the form $$(1, \boldsymbol{u}_1, 0, \boldsymbol{v}_1)\ \text{or}\ (1, \boldsymbol{u}_2, 1, \boldsymbol{v}_2),$$ where $\mathrm{wt}(\boldsymbol{u}_1)$ and $\mathrm{wt}(\boldsymbol{v}_2)$ are odd, \(\mathrm{wt}(\boldsymbol{v}_1)\) and \(\mathrm{wt}(\boldsymbol{u}_2)\) are even. The inner product of $(1, \boldsymbol{u}_1, 0, \boldsymbol{v}_1)$ with \((0, 1, 1, \dots, 1)\) equals 1. The inner product of $(1, \boldsymbol{u}_2, 1, \boldsymbol{v}_2)$ with the first row of $G_f$ equals to 1. Therefore, any codeword formed by a linear combination of the first row and the remaining rows does not belong to \(\mathcal C_f^{\perp}\). 

    (2) When \(\mathrm{wt}(f(x))\) is even, it can still be proven that for any codeword \(c_1 \in \mathcal C_f\), there exists a codeword \(c_2 \in \mathcal C_f\) such that the inner product of \(c_1\) and \(c_2\) is nonzero. The proof follows a similar process to (1).

In conclusion, \(\mathcal C_f \cap \mathcal C_f^\perp = \{0\}\).

Perform elementary row and column operations on \( G_f \): add the first row to each of the remaining rows, and swap the \((m+2)\)-th column with the first column, resulting in \( G_f' \).
$$G_f'=
    \begin{pmatrix}
          1&  &  & & &1&1&1&\cdots &1\\
          &1& & & &1&\widehat{f}_0&\widehat{f}_1&\cdots &\widehat{f}_{m-1} \\
          & &1& & &1&\widehat{f}_{m-1}&\widehat{f}_0&\cdots&\widehat{f}_{m-2}\\
          & &  &\ddots& &\vdots&\vdots&\vdots&\ddots&\vdots \\
          &  &  & &1&1&\widehat{f}_1&\widehat{f}_2&\cdots&\widehat{f}_0\\
    \end{pmatrix}.$$
    The code generated by \( G_f' \) is the bordered DC code corresponding to \( \widehat{f}(x) \). Since \( G_f' \) is equivalent to \( G_f \), the bordered DC code corresponding to \( \widehat{f}(x) \) is also an LCD code.
\end{proof}

\section{Conclusion}
In this paper, we study DC codes in the form of polynomials. We propose the conditions for DC codes and bordered DC codes over $\mathbb F_2$ to be self-dual and extremal. Using the method presented in the theorem, we can find all the extremal DC codes with length up to 20 and some cases with the length between 22 and 44. Furthermore, the sufficient conditions for a bordered DC code to be an LCD code under Euclidean inner product are presented. A further direction to consider is the simple conditions for double circulant codes to be extremal when the code length is longer or the field is larger.

\nocite{ref26}
\bibliographystyle{cas-model2-names}
\bibliography{main}

\end{document}